\documentclass[pra,aps,twocolumn,10pt,showpacs,nofootinbib,superscriptaddress]{revtex4-2}
\usepackage[utf8]{inputenc}
\usepackage[T1]{fontenc}
\usepackage{amsmath,amsfonts,amssymb,amsthm,charter}
\usepackage{mathcmd,mathrsfs,stmaryrd}
\usepackage{graphicx,xcolor} 

\usepackage{braket,bm,bbold}
\usepackage{tikz,rotating} 

\newtheorem{definition}{Definition}
\newtheorem{theorem}{Theorem}
\newtheorem{example}{Example}
\newtheorem{proposition}{Proposition}

\newcommand{\ms}{\mathrm{s}} 
\newcommand{\id}{\mathrm{id}} 
\newcommand{\Tr}{\mathrm{Tr}} 

\usepackage{ulem}
\usepackage[colorlinks=true, linkcolor=purple, citecolor=blue]{hyperref}

\begin{document}

\title{The entangling power of non-entangling channels}

\author{Julien Pinske}
     \affiliation{Niels Bohr Institute, University of Copenhagen, Jagtvej 155 A, DK-2200 Copenhagen, Denmark}
     \email{julien.pinske@nbi.ku.dk}

\author{Jan Sperling}
    \affiliation{Institute for Photonic Quantum Systems (PhoQS), Paderborn University, Warburger Stra\ss{}e 100, 33098 Paderborn, Germany}

\author{Klaus M\o lmer}
    \affiliation{Niels Bohr Institute, University of Copenhagen, Jagtvej 155 A, DK-2200 Copenhagen, Denmark}

\date{\today}

\begin{abstract}
    There are processes that cannot generate entanglement but may, nevertheless, amplify entanglement already present in a system.
    Here, we show that a non-entangling operation can increase the Schmidt number of a quantum state only if it can generate entanglement with some non-zero probability.
    This is in stark contrast to the case where the parties of a quantum network
    are only able to control their joint state by local operations and classical communication (LOCC).
    There, being able to apply operations probabilistically (stochastic LOCC) does not increase the Schmidt number.
    Our findings show that certain non-entangling operations become entangling when selecting on specific measurement outcomes.
    This naturally leads us to the class of stochastically non-entangling maps, being those that cannot generate entanglement even probabilistically.
    Intrigued by this finding, we devise a Schmidt number for quantum channels that quantifies whether a channel can generate entanglement probabilistically.
    Moreover, we show that a channel is non-entangling if and only if its dual map is \textit{witness-preserving}---it takes entanglement witnesses to witnesses.
    Based on this finding, we derive inequalities whose violation signals that a process generates entanglement.
    \end{abstract}
    
    \maketitle

    \section{Introduction}
    \label{sec:intro}

    Quantum entanglement describes correlations between two or more subsystems that cannot be attributed to any classical joint description \cite{W89}.
    Research in the past decades has spurred a strong interest in entanglement because of its foundational aspects \cite{S05,HB15} and its central role in quantum information science \cite{E91,R17}.
    However, many processes do not generate entanglement, and it is tempting to think of these as marking the limit of classical physics.
    The most general process that preserves the separability of a composite system is described by a non-entangling map \cite{HN03}.
    Such a process may correspond to a sequence of local quantum operations or describes the accumulation of joint classical noise in a network.

    For special quantum tasks, it is often useful to consider specific subclasses of non-entangling maps, such as local operations and classical communication (LOCC) \cite{PW91}.
    In the LOCC paradigm, spatially separated parties transmit classical information, but only local quantum operations are applied.
    The mathematical structure of such protocols has been the subject of extensive studies, and there are several variants of LOCCs \cite{C04,HHH09,CL14}.
    While relatively easy to implement, LOCC protocols are quite limited in what they can achieve \cite{SWG18}.
    In particular, any measure of entanglement does not increase under LOCC \cite{BDS96}. 
    The situation changes when the parties in a quantum network (post-)select a specific outcome of their local operations that occurs with some non-zero probability.
    This approach to manipulating entanglement probabilistically is known as stochastic LOCC (SLOCC) \cite{DVC0}.
    SLOCC is more powerful than LOCC \cite{I04,LMD08} and, in particular, it can increase certain entanglement measures \cite{ODT07,BCD13}.
    On the other hand, measures unaffected by this type of postselection have been devised as well, such as the Schmidt number \cite{TH00,SV11}.

    Moving beyond LOCC, subclasses of non-entangling operations have been specified with reference to their Choi matrix  \cite{C75}.
    For example, the class of separable maps \cite{BD99,R99,GG07} contains operations whose Choi matrix is separable \cite{CD01}.
    The so-called positive partial-transpose (PPT) operations \cite{R99(2),MW08}, relevant for bound entanglement \cite{BDM99} and entanglement distillation \cite{XRL10}, have a Choi matrix which is PPT.
    Non-entangling maps also include operations that remove entanglement from a given state, such as
    the heavily studied class of entanglement-breaking channels \cite{HS03,MZ10} as well as entanglement-destroying maps \cite{LH17}.
    These are of interest for studying the cost of erasing quantum correlations \cite{BB18} as well as imposing censorship on a quantum network \cite{PS24,PM24,PM25}.
    
    Notably, there are non-entangling maps that are more powerful than any of the above classes \cite{VH05,BP10,PAH24,APH24}.
    These include operations that are non-entangling but not completely non-entangling \cite{CG19}; i.e., they can distribute entanglement to previously unentangled parties when acting on parts of a larger, composite system.
    A prominent example is the swap gate which enables a protocol for entanglement swapping between nodes in a quantum network \cite{PBW98}.
    Some of these operations have been shown to be able to increase the Schmidt number, and, unlike LOCC, some non-entangling maps may be able to distill entanglement even from bound entangled states \cite{CV20}.

    Despite these advances, a coherent picture of how non-entangling dynamics acquire operational entangling power is still lacking.
    In particular, the role of probabilistic postselection—ubiquitous in experimental realizations of entanglement generation and manipulation—has not been clarified beyond the SLOCC framework.
    Closely related, a systematic, convex classification of quantum channels according to their ability to create or amplify entanglement remains missing.
    Understanding these issues is essential for identifying the boundaries of non-entangling dynamics and for explaining how certain operations, which hide latent entangling capabilities, can nevertheless be harnessed to distribute or distill entanglement.
    
    In this paper, we show that the (entangling) power of non-entangling maps originates from them being not stochastically non-entangling.
    Stochastically non-entangling maps are those operations which do not generate entanglement even probabilistically.
    This class contains within it the set of (stochastic) LOCC and separable operations, but is strictly smaller than the set of non-entangling maps; see Fig~\ref{fig:sets}. 
    That a process does not produce quantum correlations, even when allowing for postselection, turns out to be a very strict condition, as has been identified within the framework of general quantum resource theories \cite{CG19}.
    For instance, in the case of quantum coherence this marks the difference between incoherent and maximally incoherent operations \cite{BC14,CG17} while the analogous case for entanglement is a much less explored subject.
    Here, we show that stochastically non-entangling maps do not increase the Schmidt number of a state.
    This implies that any non-entangling map that can increase the Schmidt number can generate entanglement when supplied by postselection.
    This is a remarkable result, as it implies that one can make still use of the entangling power generated by a selective process even when one has only access to the non-selective operation.
    Our analysis suggests that the capability of some non-entangling maps to amplify and distill entanglement is due to one of their Kraus operators being actually entangling.
    Intrigued by this finding, we carry out a convex roof construction \cite{U98} over the set of quantum channels, which leads to what we dub the \textit{channel Schmidt number}.
    Afterwards, we identify the entire class of non-entangling processes as those whose dual map is witness-preserving; i.e., it takes entanglement witnesses \cite{B05,GT09,SV13} to witnesses.
    We obtain several useful properties of these maps and derive Bell-like inequalities \cite{T00} for the detection of entangling maps.
    Surprisingly, despite non-entangling maps being able to amplify entanglement already present in a quantum state, their dual maps do not refine a witness.

    The paper is structured as follows. 
    In Sec.~\ref{sec:SLOCC}, we review the (S)LOCC paradigm and its relation to the Schmidt number.
    Section~\ref{sec:SNE} introduces stochastically non-entangling maps, and we show that these do not increase the Schmidt-number.
    Based on their convex structure, we devise a Schmidt number for quantum channels and derive its most important properties in Sec.~\ref{sec:Schmidt-number}.
    In Sec. \ref{sec:witness}, we consider the entire class of non-entangling maps and show that these operations are characterized by their dual map preserving entanglement witnesses.
    Based on this finding, we derive inequalities whose violation signals that a process generates entanglement in Sec.~\ref{sec:Bell-like}.
    Finally, Sec.~\ref{sec:Fin} is reserved for a summary and concluding remarks.

    \begin{figure}
        \centering
        \begin{tikzpicture}
        \node at (0,0) {\includegraphics[width=0.48\textwidth]{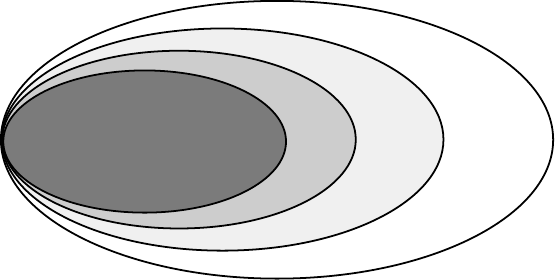}};
        \node at (-2,0.2) {LOCC/SLOCC};
        \node at (-2,-0.2) {separable maps};
        \node at (0.6,0) {$\mathcal{N}_{\ms}$};
        \node at (1.9,0.13) {non-ent.};
        \node at (1.9,-0.13) {maps};
        \node at (3.4,0.13) {quantum};
        \node at (3.4,-0.13) {channels};
        \end{tikzpicture}
        \caption{%
            Hierarchy of quantum channels.
            The set of channels (completely positive maps) comprises non-entangling maps, stochastically non-entangling maps $\mathcal{N}_\ms$, as well as (stochastic) LOCC and separable maps.
            All subsets are proper subsets.
        }\label{fig:sets}
    \end{figure}

    \section{Entanglement measures and separable maps}
    \label{sec:SLOCC}

    Consider $n$ parties that share an entangled state $\rho$, i.e.,
    \begin{equation}
        \label{eq:Sep-State}
        \rho\neq \sum_a p_a \rho^a_1\otimes\dots\otimes \rho^a_n,
    \end{equation}
    for $p_a\geq 0$.
    Due to technological limitations or the parties being far-separated from each other, they cannot manipulate the state $\rho$ arbitrarily.
    Instead, the parties apply local quantum operations to their individual systems while coordinating their action using classical communication (LOCC).
    While the precise mathematical structure of such a protocol may vary \cite{C04,HHH09,CL14}, any LOCC operation corresponds to a separable map \cite{BD99,R99,GG07},
    \begin{equation}
        \label{eq:sep-map}
        \Lambda(\rho)=\sum_{j} \big(M_1^j\otimes\dots \otimes M_n^j\big) \rho \big(M_1^j\otimes\dots \otimes M_n^j\big)^\dag.
    \end{equation}
    
    The quantum channel $\Lambda$ is deterministic in the sense that it is indifferent to the specific outcomes $j$ in Eq.~\eqref{eq:sep-map}. 
    Mathematically, this corresponds to $\Lambda$ being trace-preserving, $\Tr(\Lambda(\rho))=\Tr(\rho)$.
    In contrast, if the parties of a network are able to select a certain $j$, then they implement
    \begin{equation}
        \label{eq:stoch-sep-map}
        \Lambda_j(\rho)=\big(M_1^j\otimes\dots \otimes M_n^j\big) \rho \big(M_1^j\otimes\dots \otimes M_n^j\big)^\dag,
    \end{equation}
    with some non-zero probability $\Tr(\Lambda_j(\rho))\leq 1$.
    This approach to manipulating entanglement is known as stochastic LOCC (SLOCC), and it is more powerful than LOCC \cite{DVC0,I04,LMD08}.
    For example, the two-mode state $\ket{\psi}=(\sqrt{2}\ket{01} + \ket{12})/\sqrt{3}$ of an electromagnetic field can be transformed into the so-called maximally entangled state $\ket{\phi^+}\propto \sqrt{2/3}(\ket{00} + \ket{11})$ by applying the annihilation operator $\hat{a}$ to the second mode \cite{KPL14}. 
    Despite that, there exist entanglement measures that are non-increasing under SLOCC, such as the Schmidt number \cite{TH00,SV11}
    \begin{equation}
        \label{eq:Schmidt-for-state}
    	r(\rho)= \min_{(p_a,\psi_a)}\Big\{\max_a r(\psi_a):~\rho=\sum_a p_a\ket{\psi_a}\bra{\psi_a}\Big\},
    \end{equation}
    where one minimizes over possible mixtures.
    That is, $r(\Lambda(\rho))\leq r(\rho)$, for any $\Lambda$ of the form in Eq.~\eqref{eq:stoch-sep-map}.
    The definition in Eq.~\eqref{eq:Schmidt-for-state} includes the multi-partite Schmidt rank $r(\psi)$ of a pure state $\ket{\psi}$, being the minimal number of linearly independent terms so that \cite{EB01}
    \begin{equation}
        \ket{\psi}=\sum_{b=1}^{r(\psi)}\ket{\psi_{1}^b}\otimes \dots \otimes \ket{\psi_{n}^b}.
    \end{equation}
    Informally speaking, probabilistic operations are able to rescale the amplitudes of a quantum state due to the output state in Eq.~\eqref{eq:stoch-sep-map} not being normalized.
    While this can alter the Schmidt coefficients of a pure state, its Schmidt rank remains unaffected.
    In conclusion, within the LOCC paradigm the possibility of postselection does not enable the parties in a system to increase the Schmidt number.
    It is a primary concern of this paper to show that the situation changes when we move beyond LOCC and consider the most general class of non-entangling maps. 

    \section{Stochastically non-entangling maps}
    \label{sec:SNE}

    The most general change a system can undergo without creating entanglement is given by a non-entangling map.
    A channel $\Lambda$ is said to be non-entangling if
        \begin{equation}
            \forall \sigma\in\mathcal{S}:\quad \Lambda(\sigma)\in\mathcal{S},
        \end{equation}
    where $\mathcal{S}$ denotes the set of separable states.
    Thus, $\Lambda$ is entangling if there exists at least one separable state $\sigma\in\mathcal{S}$ such that $\Lambda(\sigma)\notin \mathcal{S}$ is entangled. 
    Non-entangling maps have the potential to increase the Schmidt number \cite{CV20}, as the following example illustrates.

    \begin{example}
        \label{ex:non-mono}
        \normalfont
        Consider the two-qudit channel
        \begin{equation}
            \label{eq:ex-non-ent}
            \Lambda(\rho)=\Tr(\phi^+_k\rho)\ket{\psi}\bra{\psi} + \Tr((\mathbb{1}^{\otimes 2}-\phi^+_k)\rho)\frac{\mathbb{1}^{\otimes 2}}{d^2},
        \end{equation}
        where $\phi^+_k=\ket{\phi^+_k}\bra{\phi^+_k}$ and $\ket{\phi^+_k}=\tfrac{1}{\sqrt{k}}\sum_{a=0}^{k-1}\ket{aa}$ has Schmidt rank $k$ and 
        \begin{equation}
            \label{eq:large-Schmidt-state}
            \ket{\psi}=\sum_{b=0}^{d-1}\lambda_b \ket{bb},
        \end{equation}
        has Schmidt coefficients $\lambda_{0}\geq\dots \geq \lambda_{d-1}>0$.
        $\Lambda$ is non-entangling for $\lambda_0\lambda_1\leq (k-1)/d^2$ (see Appendix~\ref{app:robust}).
        Note that, $\Lambda$ increases the Schmidt rank, $r(\Lambda(\phi^+_k))=d$. 
    \end{example}

    \subsection{Stochastically non-entangling maps}
    
    Despite a channel $\Lambda$ being non-entangling, any of its Kraus representations, $\Lambda(\rho)=\sum_j M_j \rho M_j^\dag$, may still have a Kraus operator $M_j$ that creates entanglement.
    That is, there exists a separable state $\sigma\in\mathcal{S}$, such that $M_j\sigma M_j^\dag\notin\mathcal{S}$ is entangled, for some $j$.
    We refer to such a channel as being non-entangling but not stochastically non-entangling. 
    Motivated by this observation, we introduce the class of stochastically non-entangling maps, being those channels that cannot create entanglement even when supplied by postselection.
    \begin{definition}
        \label{def:ne-maps}
        \normalfont 
        $\Lambda$ is stochastically non-entangling if 
        \begin{equation*}
            \Lambda = \sum_j \Lambda_j,
        \end{equation*}
        with $\Lambda_j(\rho)=M_j\rho M_j^\dag$ being non-entangling for all $j$.
    \end{definition}
    The physical significance of stochastically non-entangling channels is that they remain non-entangling even when one is able to select a specific outcome $j$.
    Clearly, any stochastically non-entangling map is non-entangling.
    Given some Kraus representation, $\Lambda(\rho)=\sum_{j}M_j \rho M_j^\dag$, it is, generally, nontrivial to decide whether $\Lambda$ is stochastically non-entangling, due to the non-uniqueness of the operators $M_j$.
    
    We denote the set of stochastically non-entangling maps by $\mathcal{N}_{\ms}$. 
    Mathematically, $\mathcal{N}_{\ms}$ is obtained from the convex closure of maps $\Lambda(\rho)=M\rho M^\dag$, whose single Kraus operator $M$ preserves any product state, i.e.,
    \begin{equation}
        \label{eq:prod-to-prod}
        M\ket{\psi_1}\otimes \dots \otimes \ket{\psi_n}=\ket{\phi_1}\otimes \dots \otimes \ket{\phi_n}.
    \end{equation}
    This includes Kraus operators of the following forms:
    \begin{equation}
    \label{eq:Kraus-ne}
    \begin{aligned}
        M_1\otimes \dots \otimes M_n,
        \quad
        \big( M_1\otimes \dots \otimes M_n\big)V_{\pi},
        \\
        \text{and}\quad
        \ket{\psi_1,\dots,\psi_n}\bra{\Psi},
    \end{aligned}
    \end{equation}
    where $V_{\pi}$ exchanges tensor factors according to a permutation $\pi$ (e.g., the swap operator), and $\ket{\Psi}$ is an arbitrary $n$-partite state.
    For $n=2$, the examples in Eq.~\eqref{eq:Kraus-ne} cover all possible operators satisfying Eq.~\eqref{eq:prod-to-prod} \cite{G14}. For $n\geq 3$, also combinations of those can be applied to different partitions of the $n$ subsystems; e.g., $M_1\otimes\dots \otimes M_{k}\otimes \ket{\psi_{k+1},...,\psi_n}\bra{\Psi}$, with $\ket{\Psi}\in\mathcal{H}^{\otimes (n-k)}$, obeys Eq.~\eqref{eq:prod-to-prod} as well.
    Moreover, every invertible operator obeying Eq.~\eqref{eq:prod-to-prod} can be written as \cite{MM59,W67}
    \begin{equation}
        \big(U_1\otimes \dots \otimes U_n \big) V_{\pi},
    \end{equation}
    with generally invertible (not just unitary) operators $U_i$ acting on the $i$th subsystem.
    The simple characterization of these channel's Kraus operators in Eq.~\eqref{eq:Kraus-ne} makes the class $\mathcal{N}_\ms$ particularly amenable to work with.
    For example, one readily arrives at the following proposition.
    \begin{proposition}
        \label{prop:rank-mono}
        \normalfont
        Let $M$ map product states onto product states.
        Then $r(M\ket{\psi})\leq r(\psi)$ for any pure state $\ket{\psi}$.
    \end{proposition}
    \begin{proof}
        Every pure state with an $n$-partite Schmidt rank $r(\psi)$, by definition, can be expanded as
        \begin{equation}
            |\psi\rangle=\sum_{a=1}^{r(\psi)} |\psi_1^a\rangle\otimes\cdots\otimes|\psi_n^a\rangle,
        \end{equation}
        where $\{|\psi_1^a\rangle\otimes\cdots\otimes|\psi_n^a\rangle\}_a$ are linearly independent.
        Since $M$ preserves product states,
        \begin{equation}
            M|\psi\rangle=\sum_{a=1}^{r(\psi)} |\phi_1^a\rangle\otimes\cdots\otimes|\phi_n^a\rangle
        \end{equation}
        has at most $r(\psi)$ linearly independent, product elements $|\phi_1^a\rangle\otimes\cdots\otimes|\phi_n^a\rangle = M(|\psi_1^a\rangle\otimes\cdots\otimes|\psi_n^a\rangle)$,
        where some superscripts $a$ and $a'\neq a$ could result in linearly dependent elements if $M$ is not invertible;
        thus, $r(M\ket{\psi})\leq r(\psi)$.
    \end{proof}
    The set $\mathcal{N}_\ms$ is strictly larger than the set of (stochastic) separable maps~\eqref{eq:stoch-sep-map}.
    For example, the channel $\Lambda(\rho)=\Tr(\phi^+_k\rho)\ket{00}\bra{00}$ is stochastically non-entangling because of its single Kraus operator $M=\ket{00}\bra{\phi^+_k}$ preserving product states; see Eq.~\eqref{eq:prod-to-prod}.
    Simultaneously, it is not separable, because $M\neq M_1\otimes M_2$ is not a simple tensor product.

    \subsection{The Schmidt-number under stochastically non-entangling maps}

    Similar to separable maps, stochastically non-entangling maps do not increase the Schmidt number.
    \begin{theorem}
        \label{th:monotone}
        Let $\Lambda$ be stochastically non-entangling.
        Then,
        \begin{equation*}
            \forall \rho: \quad r(\Lambda(\rho))\leq r(\rho).
        \end{equation*}
    \end{theorem}
    \begin{proof}
    Let $\Lambda$ be stochastically non-entangling.
    Then it has a Kraus representation, $\Lambda(\rho)=\sum_j M_j \rho M_j^\dag$, with $\Lambda_j(\rho) = M_j \rho M_j^\dag$ being non-entangling for all $j$.
    Since $r$ is convex, $r(\Lambda(\rho))\leq \sum_{j} r(\Lambda_j(\rho))$, and Proposition \ref{prop:rank-mono} yields $r(\Lambda_j(\rho))\leq r(\rho)$, we have that $r(\Lambda(\rho))\leq r(\rho)$.
    \end{proof}
    Recall that the channel $\Lambda$ in Example~\ref{ex:non-mono} can increase the Schmidt number.
    Now, Theorem~\ref{th:monotone} implies that $\Lambda$ cannot be stochastically non-entangling.
    If an experimenter is able to postselect on the outcomes $j$ of a Kraus representation, $\Lambda(\rho)=\sum_j M_j \rho M_j^\dag$, then they can generate entanglement using one of the $M_j$ which is entangling. 
    For concreteness, consider the experimenter is able to outcome-resolve the channel $\Lambda$ into
    \begin{equation}
        \label{eq:Kraus-ex-non-mono}
        \begin{split}
            M_0 & = \ket{\psi}\bra{\phi_k^+},\quad M_{jk} = \frac{1}{d}\ket{\xi_k}\bra{\zeta_j},\quad j\neq 0,\\
        \end{split}
    \end{equation}
    where $\ket{\zeta_j}$ and $\ket{\xi_k}$ form orthonormal bases, and $\ket{\zeta_0}=\ket{\phi_k^+}$.
    One readily verifies that 
    \begin{equation}
        \Lambda(\rho) = M_0 \rho M_0^\dag + \sum_{j=1,k=0}^{d^2-1}M_{jk}\rho M_{jk}^\dag
    \end{equation}
    recovers the channel in Example~\ref{ex:non-mono}.
    If an experimenter postselects on the outcome $0$, then the entangled state $M_0\sigma M_0^\dag\propto\ket{\psi}\bra{\psi}$ from Eq.~\eqref{eq:large-Schmidt-state} is generated with probability $\Tr(M_0^\dag M_0 \sigma)=\braket{\phi_k^+|\sigma|\phi_k^+}\leq 1/k$ for a separable input state $\sigma\in\mathcal{S}$.
    The specific form of these states $\ket{\zeta_j}$ and $\ket{\xi_k}$ depends on the physical origin of the white-noise contribution, $\mathbb{1}^{\otimes 2}/d^2$, in Eq.~\eqref{eq:ex-non-ent}.
    
    To give another operational viewpoint, suppose the parties of a quantum network are initially able to generate entanglement (probabilistically) using the map $\Lambda_j(\rho)=M_j\rho M_j^\dag$.
    However, due to a loss of knowledge about the measurement outcome $j$, they are confronted with having implemented the non-selective process $\Lambda(\rho)=\sum_j M_j\rho M_j^\dag$, which may well be non-entangling, caused by the statistical average.
    Despite that, the parties may still make use of the \textit{lost} (unknown) entanglement to increase the Schmidt number of their shared state.

    To gain more insight into the properties of stochastically non-entangling maps, we consider their action on states that already possess some input entanglement.
    Let $\mathcal{S}^k$ denote the set of states with Schmidt number smaller or equal to $k$.
    For example, the set $\mathcal{S}^1=\mathcal{S}$ contains separable states.
    The sequence of inclusions $\mathcal{S}^1\subseteq \dots \subseteq \mathcal{S}^k$,
    for any $k\geq 2$, imposes an ordering on the set of quantum states \cite{SV15}. 
    Given any stochastically non-entangling map $\Lambda\in\mathcal{N}_\ms$, Theorem~\ref{th:monotone} implies that the channel preserves the ordering, i.e.,
    \begin{equation}
        \label{eq:inclusion}
        \Lambda(\mathcal{S}^1)\subseteq \dots \subseteq \Lambda(\mathcal{S}^k).
    \end{equation}
    From this formal viewpoint, stochastically non-entangling maps are the structure-preserving maps for the theory of entanglement.

    \section{A Schmidt number for quantum channels}
    \label{sec:Schmidt-number}

    The set of stochastically non-entangling maps $\mathcal{N}_{\ms}$ is convex, and any convex set can be characterized through its extreme points.
    The extreme points of a convex set are those elements of the set that cannot be obtained from statistical mixture of other points in the set.
    For example, the set of extreme points $\mathcal{S}_0$ of the set of separable states $\mathcal{S}$ consists of pure product states.
    For the set $\mathcal{N}_{\ms}$, the extreme points are non-entangling rank-1 channels, i.e., $\Lambda(\rho)=M\rho M^\dag$, with $M$ being of the form~\eqref{eq:prod-to-prod}.
    Knowing the explicit form of the extremal points, cf. Eq.~\eqref{eq:Kraus-ne}, makes the set $\mathcal{N}_{\ms}$ paramount for a convex-roof construction \cite{U98,SV15} in which we firstly define a measure on rank-1 channels, and then extend the notion to arbitrary channels.
    The channel Schmidt rank of $\Lambda(\rho)=M\rho M^\dag$ is defined as
    \begin{equation}
        \label{eq:lhs}
        r(\Lambda)=\max_{\psi\in\mathcal{S}_0}r(M\ket{\psi}).
    \end{equation}
    
    For example, a rank-1 non-entangling map $\Lambda$ has only a single Kraus operator $M$, which preserves product states; see Eq.~\eqref{eq:prod-to-prod}.
    Thus, $r(\Lambda)=1$.
    As a second example, consider $\Lambda(\rho)=\Tr(\phi_k^+\rho)\phi_{k^\prime}^+$, with $\phi_{k}^+$ and $\phi_{k^\prime}^+$ having Schmidt rank $k$ and $k^\prime$, respectively. 
    From the Kraus operator $M=\ket{\phi_{k^\prime}^+}\bra{\phi_{k}^+}$, one finds $r(\Lambda)=k^\prime$.

    \subsection{Convex-roof construction for quantum channels}
	\label{ssec:convex-roof}
	
    Next, the analysis is extended to general channels.
    Since the Kraus representation of a channel is not unique, the minimum is taken over all possible decompositions, leading to a generalization of $r$ from Eq. \eqref{eq:lhs}.
    \begin{definition}
        \label{def:1s-Schmidt}
        \normalfont
        The channel Schmidt number of $\Lambda$ is
        \begin{equation}
            \label{eq:1s-Schmidt}
            r(\Lambda)=\min_{\{\Lambda_j\}_j}\Big\{\max_j r(\Lambda_j):~\Lambda=\sum_j \Lambda_j\Big\}.
        \end{equation}
    \end{definition} 
    Note that the different Kraus representations, $\Lambda_j(\rho)=M_j\rho M_j^\dag$, over which one has to minimize in Eq.~\eqref{eq:1s-Schmidt} are related by a unitary mixing, $K_j=\sum_{k} u_{jk} M_k$ such that $\Lambda(\rho) = \sum_j K_j \rho K_j^\dag$.
    Here, $u_{jk}$ are the components of an isometry, i.e., $\sum_{i} u_{ij}^* u_{ik} = \delta_{jk}$.
    Thus, determining the channel Schmidt number $r(\Lambda)$ amounts to the formidable computational problem of exploring all possible isometries.
    
    For a channel with one Kraus operator, the channel Schmidt number~\eqref{eq:1s-Schmidt} and the channel Schmidt rank~\eqref{eq:lhs} coincide.
    Moreover, due to the minimum in Definition \ref{def:1s-Schmidt}, $r$ is convex, i.e.,
    \begin{equation}
        r(\Lambda)\leq \sum_j r(\Lambda_j),
    \end{equation}
    for $\Lambda =\sum_j \Lambda_j$.
    Given a channel $\Lambda(\rho)=\sum_j M_j \rho M_j^\dag$, its dual map $\Lambda^*(\rho)=\sum_j M_j^\dag \rho M_j$ generally has a different channel Schmidt number, i.e., $r(\Lambda)\neq r(\Lambda^*)$.
    This is expected as the dual map of a non-entangling map $\Lambda$ can, in general, be entangling \cite{CV20}.
    For example, $\Lambda(\rho)=\Tr(\phi^+_k \rho)\mathbb{1}$ is non-entangling while $\Lambda^*(\rho)=\Tr(\rho)\phi^+_k$ is entangling.
    Finally, we conclude the following Theorem.
    \begin{theorem}
        \label{th:ne}
        A channel $\Lambda$ is stochastically non-entangling if and only if $r(\Lambda)=1$.
    \end{theorem}
    \begin{proof}
        If $\Lambda$ is stochastically non-entangling, then it can be expressed through non-entangling rank-1 channels $\Lambda_j$, i.e., $\Lambda=\sum_j \Lambda_j$, with $r(\Lambda_j)=1$.
        Thus, the minimum in Eq. \eqref{eq:1s-Schmidt} is satisfied.
        It follows that $r(\Lambda)=1$.

        Now suppose that $\Lambda$ has minimal channel Schmidt number, $r(\Lambda)=1$.
        Then, it can be written as $\Lambda(\rho)=\sum_j\Lambda_j(\rho)$, with $\Lambda_j(\rho)=M_j \rho M_j^\dag$, so that $\max_j r(\Lambda_j)=1$ holds true.
        Hence, $r(\Lambda_j)=1$ for all $j$. 
        Thus, $M_j$ preserves product states, and, thus, it is non-entangling.
        Finally, as $\Lambda$ is a convex sum of non-entangling rank-1 channels $\Lambda_j$, $\Lambda$ is stochastically non-entangling.
    \end{proof}

    \subsection{Comparison to entanglement-generating measures}
    \label{ssec:ent-gen-measure}

    Now, we compare the channel Schmidt number to other quantifiers, such as the so-called, entanglement-generating measures \cite{TGW14,LBL20,ZGY22}.
    Here, we are primarily interested in measures of the form
    \begin{equation}
    	\label{eq:ent-gen-Schmidt}
    	\hat{r}(\Lambda)=\max_{\sigma\in\mathcal{S}}r(\Lambda(\sigma)),
    \end{equation}
    where $r(\Lambda(\sigma))$ is the Schmidt number [Eq.~\eqref{eq:Schmidt-for-state}] of the generally mixed state $\Lambda(\sigma)$.
    Unlike the channel Schmidt number, the measure $\hat{r}$ is minimal if and only if $\Lambda$ is non-entangling, which is evident from Eq.~\eqref{eq:ent-gen-Schmidt}.
    It does not capture the potential of a channel $\Lambda$ to generate entanglement probabilistically.
    This has its origin in the fact that the channel Schmidt number was obtained from a convex-roof construction of channels while $\hat{r}$ originates from a convex-roof construction on the set of states.

    \subsection{Comparison to Choi-based measures}
    \label{ssec:Choi-measure}

    Another way to discuss the entanglement of a channel refers to the entanglement present in its Choi matrix \cite{C75}
    \begin{equation}
        J_{\Lambda}=(\id_{R_1\dots R_N}\otimes\Lambda)(\phi^+_{R_1 1}\otimes \dots \otimes \phi^+_{R_n n}),
    \end{equation}
    where $\ket{\phi^+}=\tfrac{1}{\sqrt{d}}\sum_{b=0}^{d-1}\ket{bb}$ is the maximally entangled state, $\id(\rho)=\rho$ denotes the identity channel, and $R_a$ is a $d$-dimensional reference system.

    One defines an entanglement measure on channels via the Schmidt number $r(J_{\Lambda})$ of a channel's Choi matrix $J_\Lambda$ \cite{CK06,MGM24}.
    To illustrate this, consider two $d$-dimensional systems and consider the replacement channel $\Lambda(\rho)=\Tr(\rho)\phi_k^+$, with $\ket{\phi_k^+}=\tfrac{1}{\sqrt{k}}\sum_{a}\ket{aa}$ having Schmidt rank $k\geq 2$. 
    Clearly, the channel is entangling, and its Choi matrix 
    \begin{equation}
        J_{\Lambda}=\frac{1}{d^2}\mathbb{1}_{R_1}\otimes\mathbb{1}_{R_2}\otimes \phi_k^+
    \end{equation}
    has the Schmidt number $r(J_\Lambda)=k$. 
    Compare this to the swap operation, $V\ket{ab}=\ket{ba}$.
    Despite being non-entangling, its Choi matrix 
    \begin{equation}
        J_{V}=\phi^+_{R_1 2}\otimes \phi^+_{R_2 1}
    \end{equation}
    is entangled along the $R_1 1:R_2 2$ partition and it has the Schmidt number $r(J_V)=d$. 
    
    Even though the swap operation is non-entangling, its Choi matrix has a larger Schmidt number than the entangling channel (if $k<d$), which shows that the Schmidt number has a nontrivial relation to the entanglement of a channel. 
    This is a consequence of the fact that while $V$ is non-entangling, it is not completely non-entangling \cite{CV20} in the sense that it can distribute---or reshuffle---entanglement among parts of a larger system. 
    In conclusion, the Choi-based measure $r(J_{\Lambda})$ is a measure of entanglement distribution not entanglement generation.

    Compare this to the channel Schmidt number (Definition~\ref{def:1s-Schmidt}).
    The swap operation $V$ is stochastically non-entangling, thus $r(V)=1$. 
    By contrast, the channel $\Lambda(\rho)=\Tr(\rho)\phi_k^+$, with $\phi_k^+$ having Schmidt rank $k$, gives rise to Kraus operators $M_j=\ket{\phi_k^+}\bra{\xi_j}$.
    Here, $\{\ket{\xi_j}\}_j$ can be any orthonormal basis of the bipartite system, owing to the fact that the Kraus representation is not unique. 
    Maximizing over pure product states yields
    \begin{equation}
        r(M_j)=\max_{\psi\in\mathcal{S}_0}r\big(\ket{\phi_k^+}\braket{\xi_j|\psi}\big)=k,
    \end{equation}
    for the channel Schmidt rank.
    The result is independent of the choice of basis, and we obtain $r(\Lambda)=k$.
    
    In summary, the channel Schmidt number is a measure able to identify the potential of a channel to create entanglement.
    It neither depends on the Choi matrix nor on any chosen notion of distance.
    It only uses the algebraic convexity of the set of stochastically non-entangling maps $\mathcal{N}_\ms$.
    It unambiguously resolves the shortcoming of the Schmidt number of the Choi matrix at distinguishing entangling and non-entangling maps.

    \section{Entanglement witnesses under dual maps of non-entangling maps}
    \label{sec:witness}

    So far, our analysis of non-entangling maps was based on their action on quantum states.
    In contrast, the dual map $\Lambda^*$ of a channel $\Lambda$ is defined by $\Tr(\Lambda^*(O) \rho)=\Tr(O\Lambda(\rho))$ and describes changes in a physical observable $O$. 
    While one cannot assign a notion of separability to a general observable (Hermitian operator), non-entangling maps can still be characterized via their dual maps.
    In this section, we show that a channel $\Lambda$ is non-entangling if and only if its dual map $\Lambda^*$ takes witnesses to witnesses.
    \begin{definition}
        \normalfont
        A Hermitian operator $W$ is a witness if
        \begin{equation}
    	    \label{eq:wit-def}
            \forall \sigma\in\mathcal{S}:\quad\braket{W}_\sigma\geq 0,
        \end{equation}
        where $\braket{W}_\sigma=\Tr(W\sigma)$ denotes the expectation value.
    \end{definition}
    A witness is called proper if there is at least one $\rho$ such that $\braket{W}_\rho< 0$ certifies its entanglement.
    If a witness is improper it is said to be trivial, and every positive semidefinite operator $W\geq 0$ is a trivial witness.
    For $W$ to be a witness, it suffices to verify Eq. \eqref{eq:wit-def} for pure product states.
    Thus, the set of witnesses is
    \begin{equation}
    	\mathcal{W}=\{W\,|\,\forall \ket{\psi}\in\mathcal{S}_0:\braket{\psi|W|\psi}\geq 0\}.
    \end{equation}
    The following theorem identifies non-entangling maps as precisely those, whose dual maps are \textit{witness preserving}.
    \begin{theorem}
        A channel $\Lambda$ is non-entangling if and only if 
        \begin{equation}
            \forall W\in\mathcal{W}:\quad\Lambda^*(W)\in\mathcal{W}.
        \end{equation}
    \end{theorem}
    \begin{proof}
        Let $\sigma\in\mathcal{S}$ be separable and let $\Lambda$ be non-entangling.
        It holds that $\braket{W}_{\Lambda(\sigma)}\geq 0$, because $\Lambda(\sigma)$ is separable.
        Therefore, $\braket{W}_{\Lambda(\sigma)}=\braket{\Lambda^*(W)}_\sigma\geq 0$ implies that $\Lambda^*(W)$ is a witness.
        
        Conversely, let $\Lambda^*$ be a channel such that $\Lambda^*(W)$ is a witness for all $W\in\mathcal{W}$.
        This implies $\braket{\Lambda^*(W)}_\sigma\geq 0$ for any separable state $\sigma$,
        thus $\braket{W}_{\Lambda(\sigma)}\geq 0$. 
        By way of contradiction, suppose $\Lambda$ was entangling.
        Then, there exist a separable state $\sigma^\prime$ such that $\Lambda(\sigma^\prime)$ is entangled.
        Since for every entangled state $\Lambda(\sigma^\prime)$ there exists a witness $W^\prime$ detecting the state's entanglement, we must have that $\braket{W^\prime}_{\Lambda(\sigma^\prime)}<0$ holds true.
        However, this is in contradiction with the previous observation that $\braket{W}_{\Lambda(\sigma)}\geq 0$.
        Thus, $\Lambda$ must be non-entangling, proving the assertion.
    \end{proof}

    From this viewpoint, non-entangling maps are those whose dual map $\Lambda^*$ takes witnesses to witnesses. 
    An important subset of non-entangling maps are entanglement-annihilating channels \cite{MZ10,FRZ12}.
    Formally, $\Lambda$ is entanglement annihilating if $\Lambda(\rho)$ is separable for all states $\rho$.
    The dual map $\Lambda^*$ of an entanglement-annihilating channel produces a trivial witness, as summarized by the following proposition.
    \begin{proposition}
        \label{prop:EA}
        \normalfont
        A channel $\Lambda$ is entanglement annihilating if and only if $\Lambda^*(W)\geq 0$ for all witnesses $W$.
    \end{proposition}
    \begin{proof}
        Let $\Lambda$ be entanglement annihilating. 
        Then, $\Lambda(\rho)$ is separable for all $\rho$. 
        Hence, 
        \begin{equation}
            \braket{\Lambda^*(W)}_\rho=\Tr\{\Lambda(\rho) W\}\geq 0,
        \end{equation}
        where $W$ is any witness.
        Thus, $\Lambda^*(W)\geq 0$ is trivial.

        Now, let $\Lambda$ output trivial witnesses, i.e., $\braket{\Lambda^*(W)}_{\rho}\geq 0$ for all $\rho$.
        This implies $\braket{W}_{\Lambda(\rho)}\geq 0$.
        Since $W$ can be any witness, it must hold that $\Lambda(\rho)$ is separable. 
    \end{proof}

    \subsection{Finer and optimal witnesses}
    \label{ssec:finer_wit}

    Formally, a witness $W$ is finer than a witness $W^\prime$ if $W$ detects all entangled states which $W^\prime$ detects and at least one additional state;
    i.e., there is $\rho\notin\mathcal{S}$ such that
    \begin{equation}
        \label{eq:finer}
    	\braket{W}_{\rho}<0\leq \braket{W^\prime}_{\rho}.
    \end{equation}
    A witness $W$ is called optimal if there is no finer witness \cite{LK00,DL14}. 
    A witness $W$ is optimal whenever it is an extreme point of the convex set $\mathcal{W}$ \cite{SR17};
    i.e., there is a product state $\ket{\psi}\in\mathcal{S}_0$ such that $\braket{\psi|W|\psi}=0$.
    If there is a separable state $\sigma\in\mathcal{S}$, satisfying $\braket{W}_\sigma=0$, then there is also a pure product state $\ket{\psi}\in\mathcal{S}_0$ for which $\braket{W}_\psi=0$.
    We find that the dual of a non-entangling map cannot make a witness optimal.
    \begin{proposition}
        \label{prop:optimal}
        \normalfont
        Let $W$ be a witness and let $\Lambda$ be non-entangling so that $\Lambda^*(W)$ is optimal.
        Then $W$ is optimal.
    \end{proposition}
    \begin{proof}
        Because $\Lambda^*(W)$ is optimal, there exists a pure product state $\ket{\psi}\in\mathcal{S}_0$ for which $\braket{\Lambda^*(W)}_\psi=0$.
        Therefore, $\braket{W}_{\Lambda(\psi)}=0$. 
        By convexity of $\mathcal{W}$, there also exists a pure product state $\ket{\phi}\in\mathcal{S}_0$ such that $\braket{W}_{\phi}=0$.
        Thus, $W$ is an extreme point of $\mathcal{W}$ and $W$ is optimal.
    \end{proof}

    \begin{figure}
        \centering
        \begin{tikzpicture}
        \node at (0,0) {\includegraphics[width=0.4\textwidth]{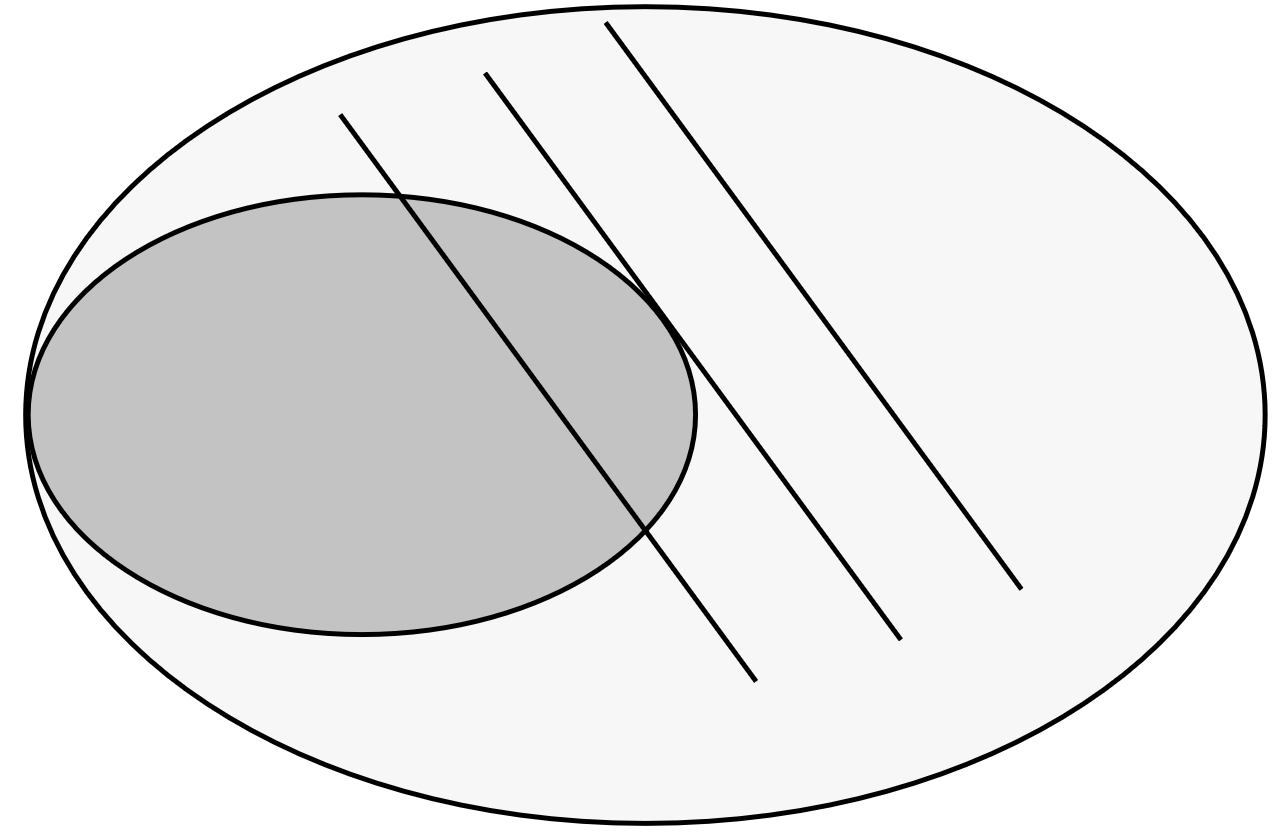}};
        \node at (-1.6,0) {$\mathcal{S}$};
        \node at (2.2,-1.2) {\footnotesize{$\Lambda^*(W)$}};
        \node at (1.4,-1.4) {\footnotesize{$W$}};
        \node at (0.5,-1.7) {\footnotesize{$\Phi^*(W)$}};
        \end{tikzpicture}
        \caption{%
            Separable states $\mathcal{S}$ (dark gray) as a subset of quantum states (light gray).
            An optimal witness $W$ corresponds to a hyperplane tangent to $\mathcal{S}$.
            The dual of a non-entangling map $\Lambda$ gives a witness $\Lambda^*(W)$ which is not finer than $W$.
            The dual of an entangling channel $\Phi$ can yield a non-witness $\Phi^*(W)$.
        }\label{fig:hyper}
    \end{figure}

    The dual of a non-entangling map cannot make a witness finer.
    To see this, note that any witness $W$ can be written as $W=\lambda\mathbb{1}^{\otimes n} - L$, with $L\geq 0$ being a generic test operator. 
    The condition for detecting entanglement then becomes $\braket{L}_{\rho}>\lambda$.
    This implies that a witness $W$ is finer than $W^\prime=\lambda^\prime\mathbb{1}^{\otimes n} - L$, whenever $\lambda^\prime> \lambda$.
    Then, $W$ is optimal for \cite{SV13}
    \begin{equation}
        \label{eq:opt-para}
        \lambda_{\mathrm{min}}=\max_{\sigma\in\mathcal{S}}\braket{L}_\sigma,
    \end{equation}
    which is the smallest $\lambda$ for which $W$ is still a witness.
    Next, consider a non-entangling map $\Lambda$ that realizes a deterministic operation; i.e., $\Lambda$ is trace preserving, $\Tr\{\Lambda(\rho)\}=\Tr\{\rho\}$.
    Its dual map $\Lambda^*$ is unital, i.e., $\Lambda^*\big(\mathbb{1}^{\otimes n}\big)=\mathbb{1}^{\otimes n}$.
    Under this map, a witness becomes $\Lambda^*(W)=\lambda\mathbb{1}^{\otimes n}-\Lambda^*(L)$.
    The non-entangling map
    \begin{equation}
        \label{eq:sep-mix}
        \Lambda(\rho)=p\rho+(1-p)\Tr(\rho)\sigma,
    \end{equation}
    amounts to mixing with a separable state $\sigma\in\mathcal{S}$.
    Its dual map reads
    \begin{equation}
        \Lambda^*(L)=pL+(1-p)\braket{L}_{\sigma}\mathbb{1}^{\otimes n},
    \end{equation}
    and we get 
    \begin{equation}
        \begin{split}
            \Lambda^*(W)&=\lambda\mathbb{1}^{\otimes n} - pL - (1-p)\braket{L}_{\sigma}\mathbb{1}^{\otimes n},\\
            &=p\left(\frac{\lambda}{p}-\frac{1-p}{p}\braket{L}_{\sigma}\right)\mathbb{1}^{\otimes n} -pL.\\
        \end{split}
    \end{equation}
    As the overall factor $p$ is irrelevant for judging whether a witness is finer, one notices that the witness $\Lambda^*(W)$ has a shifted test operator,
    \begin{equation}
        L\mapsto L + \frac{1-p}{p}\braket{L}_{\sigma}\mathbb{1}^{\otimes n}.
    \end{equation}
    Similarly, 
    \begin{equation}
        \label{eq:shift}
        \lambda \mapsto \frac{\lambda}{p}-\frac{1-p}{p}\braket{L}_{\sigma}\geq \lambda
    \end{equation}
    holds true, because of $\braket{L}_{\sigma}\leq \lambda_{\mathrm{min}}\leq \lambda$.
    Thus, $\Lambda^*(W)$ is not finer than $W$.
    For completeness, if the channel $\Lambda$ is not of the form~\eqref{eq:sep-mix}, then the witness $W$ and $\Lambda^*(W)$ cannot be compared; i.e., $\Lambda^*(W)$ is not finer than $W$ and vice versa.
    In summary, we proved the following theorem.
    \begin{theorem}
    	\label{th:not-finer}
        Let $\Lambda$ be non-entangling and trace preserving and let $W$ be a witness. 
        Then, $\Lambda^*(W)$ is not finer than $W$.
    \end{theorem}
    \hfill\qedsymbol
    
    From a geometric viewpoint, entanglement witnesses correspond to hyperplanes separating a subset of entangled states from separable ones \cite{HHH09}.
    In this sense, witnesses describe covectors (dual vectors) to state vectors $\rho$.
    Theorem \ref{th:not-finer} shows that the dual map of a non-entangling map cannot correspond to a parallel displacement of a hyperplane towards the set of separable states, Fig. \ref{fig:hyper}.
    In contrast, consider the channel $\Phi(\rho)=p\rho+(1-p)\omega$, with $\omega\notin \mathcal{S}$ being entangled.
    If $\omega$ is detected by $L$, i.e., $\braket{L}_\omega>\lambda$, its dual map $\Phi^*$ creates a shifted operator $\Phi^*(W)=\lambda^\prime \mathbb{1}^{\otimes n}-L$, with
    \begin{equation}
        \lambda^\prime = \frac{\lambda}{p}-\frac{1-p}{p}\braket{L}_{\omega}< \lambda.
    \end{equation}
    For $\lambda_{\mathrm{min}}\leq \lambda^\prime <\lambda$, $\Phi^*(W)$ is finer than $W$.
    For $\lambda^\prime \leq \lambda_{\mathrm{min}}$, $\Phi^*(W)$ is not even a witness anymore because it intersects the set of separable states, cf. Fig. \ref{fig:hyper}.

    \subsection{Schmidt-number witnesses}
    \label{ssec:Schmidt-witness}
    
    Even though non-entangling maps can increase the entanglement present in a specific state (Example~\ref{ex:non-mono}), their dual map fails to make a witness $W$ finer (Theorem \ref{th:not-finer}).
    This marks a stark difference between quantifying entanglement via the Schmidt number and approaches using witnesses.
    To further explore this apparent discrepancy, we consider so-called Schmidt-number witnesses \cite{SBL01}. 
    \begin{definition}
        \normalfont
        A witness $W_k$ is of Schmidt class $k$ iff
        \begin{equation}
            \begin{split}
                \forall \sigma \in \mathcal{S}^{k-1}:& \quad \braket{W_k}_{\sigma}\geq 0,\\
                \exists \rho \in \mathcal{S}^{k}:&\quad \braket{W_k}_\rho<0.\\
            \end{split} 
        \end{equation}
    \end{definition}
    For example, $W_3=\tfrac{2}{3}\mathbb{1}^{\otimes 2}-\phi^+_3$ is of Schmidt class $3$.
    Now consider a non-entangling map $\Lambda$ that increases the Schmidt number of a state $\sigma\in\mathcal{S}^{k-1}$ from $k-1$ to $k$, i.e., $\Lambda(\sigma)\in\mathcal{S}^k$.
    Such channels exist, cf. Example~\ref{ex:non-mono} as well as Ref. \cite{CV20}.
    For such a channel, it follows that
    \begin{equation}
        \braket{\Lambda^*(W_k)}_{\sigma}=\braket{W_k}_{\Lambda(\sigma)}<0. 
    \end{equation}
    Hence, $\Lambda^*(W_k)$ is not a witness of Schmidt class $k$.
    We stress that this finding does not contradict Theorem \ref{th:not-finer}.
    Even though, $\Lambda^*(W_k)$ detects a state which is not detected by $W_k$, the witnesses are not comparable in the sense of one being finer than the other.
    From a geometric viewpoint, the hyperplanes of $W_k$ and $\Lambda^*(W_k)$ would intersect at some point; see Fig.~\ref{fig:hyper}.
    This is so because $\Lambda$ may increase the Schmidt number, but it cannot simultaneously increase all Schmidt coefficients.
    For this case, Theorem \ref{th:not-finer} rather implies that there is no non-entangling map $\Lambda$ such that $\Lambda(\phi^+_{k-1})=\phi^+_{k}$.

    \subsection{Positive but not completely-positive maps}
    \label{ssec:PNCP}
    
    The fact that the dual of a non-entangling map $\Lambda$ cannot make a witness finer (Theorem \ref{th:not-finer}) is related to $\Lambda$ being a completely positive map.
    Notably, one could relax Definition \ref{def:ne-maps} of a non-entangling map to arbitrary positive maps. 
    For example, the partial transpose $\rho^{\Gamma}$ of a state $\rho$ is a positive, but not completely-positive map \cite{P96,HH96}.
    For every separable state $\sigma\in\mathcal{S}$, we have $\sigma^\Gamma\in\mathcal{S}$ is again separable.
    In this generalized setting, $\Lambda(\rho)=\rho^\Gamma$ may be viewed as a non-entangling map, despite $\Lambda(\rho)$ being, in general, not a quantum state when $\rho$ is entangled. 
    Notably, $\Lambda$ transforms the trivial witness $\ket{\psi}\bra{\psi}$, for $\ket{\psi}\notin\mathcal{S}$ being entangled, into an optimal witness, $\ket{\psi}\bra{\psi}^{\Gamma}$.
    Clearly, the premises of Theorem \ref{th:not-finer} are not met, and the theorem is not applicable.  

    \section{Bounds on non-entangling maps}
    \label{sec:Bell-like}

    Consider an experimenter who prepares a separable state $\sigma$ that undergoes a process $\Lambda$. 
    Subsequently, $\Lambda(\sigma)$ is subject to a measurement of a witness $W$.
    If the average measurement outcome is negative, $\braket{W}_{\Lambda(\sigma)}<0$, then $\Lambda$ is verified to be entangling.
    In the following, we illustrate this procedure for both measurement channels and random unitary channels, describing complementary concepts \cite{BSS16}.

    \subsection{Entanglement generation by measurement channels}
    \label{ssec:meas-channels}

    Consider the class of so-called measurement channels,
    \begin{equation}
        \label{eq:meas-channel}
        \Lambda(\rho) = \sum_j\Tr(E_j \rho)\rho_j,
    \end{equation}
    where $\rho_j$ are quantum states and $E_j\geq 0$ are measurement effects forming a positive operator-valued measure (POVM), i.e., $\sum_j E_j=\mathbb{1}^{\otimes n}$. 
    The channel prepares the state $\rho_j$ with probability $\Tr(E_j \rho)$.
    The dual map of $\Lambda$ is 
    \begin{equation}
        \Lambda^*(\rho)=\sum_j \Tr(\rho_j \rho) E_j.
    \end{equation}
    If $\Lambda$ is non-entangling, the action of $\Lambda^*$ on a witness $W$ leads to the inequality
    \begin{equation}
        \label{eq:WitIneq}
        \braket{\Lambda^*(W)}_{\sigma}=\sum_j \braket{E_j}_{\sigma} \braket{W}_{\rho_j}\geq 0,
    \end{equation}
    for any separable state $\sigma\in\mathcal{S}$.
    In contrast, if Eq. \eqref{eq:WitIneq} is violated for some $\sigma\in\mathcal{S}$, we conclude that $\Lambda$ is entangling.
    Then, $W$ detects $\Lambda$'s entangling power.
    Clearly, $\Lambda$ is non-entangling if all $\rho_j$ are separable, i.e., $\braket{W}_{\rho_j}\geq 0$. 
    
    For concreteness, set $n=2$ and consider states
    \begin{equation}
        \label{eq:meas-ex}
        \begin{split}
            \rho_{0} & = p \ket{\psi^{-}}\bra{\psi^{-}} + (1-p) \frac{1}{4}\mathbb{1}^{\otimes 2},\\
            \rho_{1} & = q \ket{\psi^{-}}\bra{\psi^{-}} + (1-q)\ket{\psi^{+}}\bra{\psi^{+}},\\
        \end{split}
    \end{equation}
    where $\ket{\psi^{\pm}}=(\ket{01}\pm\ket{10})/\sqrt{2}$ and $0\leq p,q\leq 1$ are probabilities, while the measurement operators are
    \begin{equation}
        \label{eq:effect-ex}
        E_{0}  = \ket{\psi^{-}}\bra{\psi^{-}},\quad E_{1}=\mathbb{1}^{\otimes 2}-\ket{\psi^{-}}\bra{\psi^{-}}.
    \end{equation}
    Further, consider the swap operator $V$ as the witness.
    Figure \ref{fig:ent-test} shows a plot of
    \begin{equation}
        \label{eq:minimize}
        \braket{\Lambda^*(V)}_{\mathrm{min}}=\min_{\sigma\in\mathcal{S}}\braket{\Lambda^*(V)}_{\sigma},
    \end{equation}
    as a function of $p$ and $q$.
    Details on the minimization in Eq.~\eqref{eq:minimize} are given in Appendix \ref{app:mini}.
    The darker shaded regions in Fig.~\ref{fig:ent-test} mark a negativity in the mean value, verifying that $\Lambda$ is entangling for these values of $p$ and $q$.
    Notice the kink in the $\braket{\Lambda^*(V)}_{\mathrm{min}}=0$ contour line at $(q,p)=(\tfrac{1}{2}, \tfrac{1}{3})$, at which both $\rho_0$ and $\rho_1$ are \textit{barely} separable. 
    For larger values of $p$ and $q$, $\Lambda$ is detected to be entangling because of the entanglement in $\rho_0$ or $\rho_1$.
    Finally, we stress that one cannot conclude whether $\Lambda$ is entangling from its Choi matrix $J_{\Lambda}$ being entangled as this merely signify that the channel is not separable; i.e., its Kraus operators cannot be written as simple tensor products; recall Sec.~\ref{ssec:Choi-measure} and that the swap channel $V$ is non-entangling but $J_V$ is entangled.
    This further motivates the treatment via entanglement witnesses, devised here.

    \begin{figure}
        \centering
        \begin{tikzpicture}
        \node at (0,0) {\includegraphics[width=0.52\textwidth]{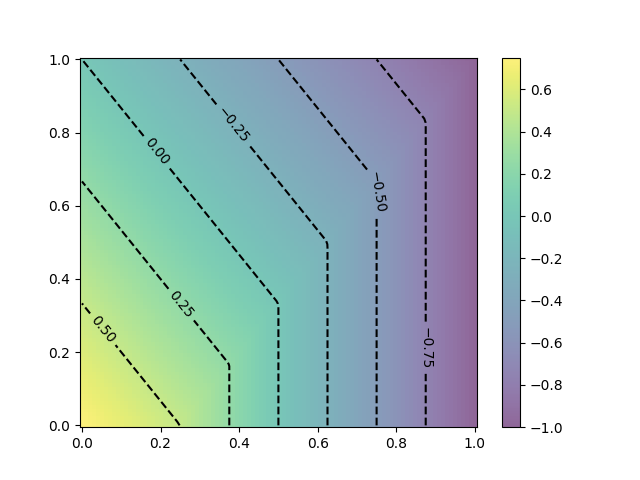}};
        \node at (-0.6,2.9) {Entanglement test};
        \node at (3,2.9) {\footnotesize{$\braket{\Lambda^*(V)}_{\mathrm{min}}$}};
        \node at (-0.6,-3.1) {$q$};
        \node[rotate=0] at (-3.9,0) {$p$};
        \end{tikzpicture}
        \caption{%
            Entanglement test in Eq.~\eqref{eq:minimize} for the witness $\Lambda^*(V)$, with the measurement channel $\Lambda$ in Eq.~\eqref{eq:meas-channel} with operators given in Eqs.~\eqref{eq:meas-ex} and~\eqref{eq:effect-ex}.
            The plot shows the minimum of $\braket{\Lambda^*(V)}_{\sigma}$ as a function of the probabilities $p$ and $q$.
        }\label{fig:ent-test} 
    \end{figure}

    \subsection{Entanglement generation by random unitary channels}
    \label{ssec:random-unitary}

    Another important class of quantum channels are random unitary channels \cite{AS08},
    \begin{equation}
        \label{eq:rand-uni}
        \Lambda(\rho)=\sum_a p_a U_a \rho U_a^\dag.
    \end{equation}
    Here the unitary $U_a$ is applied with probability $p_a\geq 0$.
    Its dual map reads 
    \begin{equation}
        \Lambda^*(\rho)=\sum_a p_a U_a^\dag\rho U_a.
    \end{equation}
    If $\Lambda$ is non-entangling, the action of $\Lambda^*$ on a witness $W$ yields the inequality
    \begin{equation}
        \label{eq:Bell-rand-uni}
        \sum_a p_a\braket{U_a^\dag W U_a}_{\sigma}\geq 0,
    \end{equation}
    which is satisfied by any separable state $\sigma\in\mathcal{S}$.
    If Eq. \eqref{eq:Bell-rand-uni} is violated by some $\sigma\in\mathcal{S}$, then $\Lambda$ is entangling.

    For illustration, let $n=2$ and consider unitaries 
    \begin{equation}
        \label{eq:ex-rand-uni}
        U_0 = \mathbb{1}^{\otimes 2},\quad U_p=\mathrm{C}X,\quad U_q=\mathbb{1}\otimes X,
    \end{equation}
    occurring with probability $(1-p-q)$, $p$, and $q$, respectively.
    Here, $X=\ket{0}\bra{1}+\ket{1}\bra{0}$ denotes the Pauli-$X$ matrix and 
    $\mathrm{C}X=\ket{0}\bra{0}\otimes \mathbb{1}+\ket{1}\bra{1}\otimes X$ is the controlled-$X$ gate.
    Further consider the witness $W=\tfrac{4}{5}\mathbb{1}^{\otimes 2}-\phi^+$. 
    Figure \ref{fig:ent-test-rand} shows a plot of $\braket{\Lambda^*(W)}_{\mathrm{min}}$ as a function of $p$ and $q$.
    The darker shaded regions mark a negativity in the mean value.
    For $q+p\leq 1$, the map $\Lambda$ is a quantum channel and thus a negativity below the diagonal of Fig.~\ref{fig:ent-test-rand} 
    verifies that $\Lambda$ is entangling for these values of $p$ and $q$. 

    \begin{figure}
        \centering
        \begin{tikzpicture}
        \node at (0,0) {\includegraphics[width=0.52\textwidth]{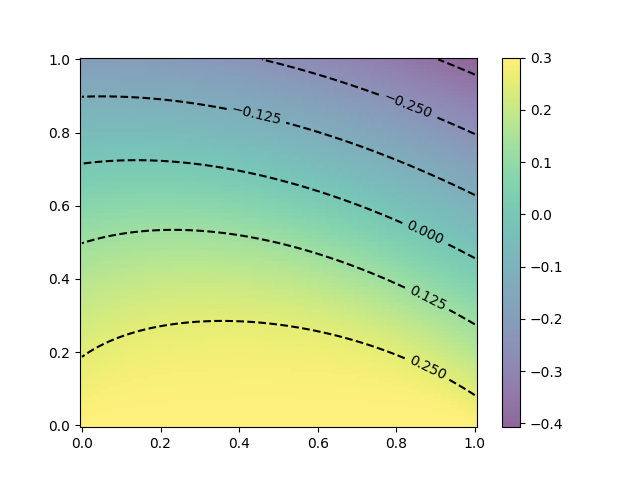}};
        \node at (-0.6,2.9) {Entanglement test};
        \node at (3,2.9) {\footnotesize{$\braket{\Lambda^*(W)}_{\mathrm{min}}$}};
        \node at (-0.6,-3.1) {$q$};
        \node[rotate=0] at (-3.9,0) {$p$};
        \end{tikzpicture}
        \caption{%
            Entanglement test for the witness $\Lambda^*(W)$ with the random unitary channel $\Lambda$ in Eq.~\eqref{eq:rand-uni} with unitaries given in Eq.~\eqref{eq:ex-rand-uni}.
            The plot shows the minimum of $\braket{\Lambda^*(W)}_{\sigma}$ as a function of the probabilities $p$ and $q$.
        }\label{fig:ent-test-rand} 
    \end{figure}

    \section{Conclusion}
    \label{sec:Fin}
    
    We have studied the potential of a quantum physical process to generate entanglement. 
    The operations that are unable to do so, even probabilistically, are stochastically non-entangling maps.
    Based on their convex structure, the channel Schmidt number was introduced, establishing a genuine measure of a channel's ability to create entanglement.
    We found that stochastically non-entangling maps do not increase the Schmidt-number while certain non-entangling maps can do so.
    This shows that probabilistic transformations can increase the Schmidt number, a feature completely absent in the (S)LOCC paradigm.
    With our theory in place, pursuing similar constructions for other convex classes of operations, such as PPT-preserving channels \cite{APE03}, may be a fruitful endeavor for future research.
    Moreover, we identified non-entangling maps as those having a witness-preserving dual map.
    There we saw that the dual map of a non-entangling process cannot make a witness finer.
    This highlights the conceptual difference between entanglement measures and witnesses.
    We further derived witness-based inequalities whose violation detects entangling dynamics, and we applied the theory to generalized measurements and random unitary channels.

    Our approach characterizes the quantum nature of channels by comparison to the classical correlations introduced by statistical mixing of channels. 
    We pointed out the deep connection between non-entangling maps and entanglement witnesses.
    This insight is useful for distinguishing quantum channels which truly produce entanglement from those who simply harness and amplify entanglement already present in a system, including its surrounding.
    In a broader context, this increases our general understanding of the entangling strength of quantum processes in physical systems.

	\acknowledgements
	We gratefully acknowledge financial support from Denmarks Grundforskningsfond (DNRF 139, Hy-Q Center for
	Hybrid Quantum Networks) and the Alexander von Humboldt
	Foundation (Feodor Lynen Research Fellowship).
    J.S. acknowledges funding through the Ministry of Culture and Science of the State of North Rhine-Westphalia (PhoQC initiative) and the QuantERA project QuCABOoSE.

	\appendix

    \section{Proof that $\Lambda$ in Example \ref{ex:non-mono} is non-entangling}
	\label{app:robust}
    
    Consider the two-qudit channel
    \begin{equation}
        \label{eq:channel-general}
        \Lambda(\rho)=\Tr\{\phi_k^+\rho\}\ket{\psi}\bra{\psi} + \Tr\{(\mathbb{1}-\phi_k^+)\rho\}\frac{\mathbb{1}^{\otimes 2}}{d^2},
    \end{equation}
    where $\ket{\phi^+_k}=\tfrac{1}{\sqrt{k}}\sum_{a=0}^{k-1}\ket{aa}$ has Schmidt rank $k$ and 
        \begin{equation}
            \ket{\psi}=\sum_{b=0}^{d-1}\lambda_b \ket{bb},
        \end{equation}
    has Schmidt coefficients $\lambda_{0}\geq\dots \geq \lambda_{d-1}>0$.
    Here, we show that $\Lambda$ is non-entangling for $\lambda_0\lambda_1\leq (k-1)/d^2$.
    First, the state 
    \begin{equation}
        \Lambda(\sigma)=p\ket{\psi}\bra{\psi} + (1-p)\frac{\mathbb{1}^{\otimes 2}}{d^2}
    \end{equation}
    is separable for \cite{VT99}
    \begin{equation}
        p\leq \frac{1}{1+d^2\lambda_0\lambda_1},
    \end{equation}
    where $\lambda_0,\lambda_1$ are the largest Schmidt coefficients of $\ket{\psi}$.
    Second, for any separable state $\sigma\in\mathcal{S}$, one has $\Tr(\phi_k^+\sigma)\leq 1/k$.
    Thus, the state $\Lambda(\sigma)$ is separable for
    \begin{equation}
        \label{eq:A1}
        \lambda_0\lambda_1\leq (k-1)/d^2.
    \end{equation}
    For Schmidt coefficients satisfying Eq.~\eqref{eq:A1} $\Lambda$ is non-entangling; see also Ref.~\cite{CV20} for similar examples.

    \section{Minimization of mean values}
	\label{app:mini}
	
	In this appendix, we give additional details on the optimization problems encountered in sections \ref{ssec:meas-channels} and \ref{ssec:random-unitary}.
    The optimization is of the form
    \begin{equation}
        \min_{\sigma\in\mathcal{S}} \braket{W}_{\Lambda(\sigma)}=\min_{\sigma\in\mathcal{S}} \braket{\Lambda^*(W)}_{\sigma},
    \end{equation}
    where $W$ is a witness and $\Lambda$ is a channel.
	First, note that it suffices to minimize over pure product states.
	To see this, let $\sigma=\sum_ap_a\ket{\chi_a}\bra{\chi_a}$ be any separable state, with 
    \begin{equation}
        \ket{\chi_a}=\ket{\chi_1^a}\otimes \dots\otimes\ket{\chi_n^a}\in\mathcal{S}_0,
    \end{equation}
    being pure product states.
    Without loss of generality, assume $\ket{\chi_1}$ is the pure product state of the ensemble, which minimizes the mean value, i.e., $\braket{\Lambda^*(W)}_{\chi_1}\leq \braket{\Lambda^*(W)}_{\chi_a}$, for all $a$.
	Then,
	\begin{equation}
		\begin{split}
			\braket{\Lambda^*(W)}_\sigma&=\sum_a p_a\braket{\Lambda^*(W)}_{\chi_a}\\
			&\geq  \sum_a p_a \braket{\Lambda^*(W)}_{\chi_1}\\
			&=\braket{\Lambda^*(W)}_{\chi_1}.
		\end{split}
	\end{equation}
	Thus, for every separable state $\sigma\in\mathcal{S}$, there is a pure product state $\chi\in\mathcal{S}_0$ of lower or equal mean value.
	Therefore, we obtain the sought result 
	\begin{equation}
		\min_{\sigma\in\mathcal{S}} \braket{\Lambda^*(W)}_\sigma=\min_{\chi\in\mathcal{S}_0} \braket{\Lambda^*(W)}_\chi,
	\end{equation}
    showing that optimization over product states suffices.

    \begin{widetext}
        For the bipartite measurement channel $\Lambda$ [see Eq. \eqref{eq:meas-channel}] as specified in Eq. \eqref{eq:meas-ex}, we get 
        \begin{equation}
        \begin{split}
             \braket{\Lambda^*(V)}_{\chi}
            =&\braket{\psi^{-}}_{\chi}\left(p\braket{V}_{\psi^-}+(1-p)\braket{V}_{\tfrac{1}{4}\mathbb{1}^{\otimes 2}}\right)
            + \braket{\mathbb{1}^{\otimes 2}-\psi^{-}}_{\chi}\left(q\braket{V}_{\psi^-}+(1-q)\braket{V}_{\psi^+}\right)\\
            =& \frac{1-3p}{2}\braket{\psi^{-}}_{\chi} + (1-2q)\big(1-\braket{\psi^{-}}_{\chi}\big),
        \end{split}
        \end{equation}
        where $V$ is the swap operation and $\braket{V}_{\psi^{\pm}}=\pm1$, for $\psi^{\pm}=\ket{\psi^\pm}\bra{\psi^\pm}$.
        Thus, $\braket{\psi^{-}}_{\chi}$ is the only quantity over which we need to optimize, and it takes values within the interval $[0,1/2]$.
    
        Next, we consider the random unitary channel $\Lambda$ [see Eq. \eqref{eq:rand-uni}] as specified in Eq. \eqref{eq:ex-rand-uni},
	    \begin{equation}
        \label{eq:min-rand-uni}
		\begin{split}
			\braket{\Lambda^*(W)}_{\chi}&=\frac{4}{5}-\braket{\Lambda^*(\phi^+)}_{\chi},\\
            &=\frac{4}{5}-(1-p-q)\braket{\phi^+}_\chi - p\braket{\mathrm{C}X \phi^+ \mathrm{C}X}_\chi - q\braket{(\mathbb{1}\otimes X) \phi^+ (\mathbb{1}\otimes X)}_\chi,\\
			&=\frac{4}{5}-(1-p-q)\braket{\phi^+}_\chi - p\braket{\ket{+0}\bra{+0}}_\chi- q\braket{\psi^+}_\chi,\\
		\end{split}
	    \end{equation}
	    where $W=\tfrac{4}{5}\mathbb{1}^{\otimes 2}-\phi^+$ is the witness and we used that $\Lambda^*$ is unital, $\Lambda^*(\mathbb{1}^{\otimes 2})=\mathbb{1}^{\otimes 2}$.
    	For a product state
        \begin{equation}
            \label{eq:pure-prod-state}
            \ket{\chi}=\big(\alpha_0\ket{0}+\alpha_1e^{i\varphi}\ket{1}\big)\otimes \big(\beta_0\ket{0}+\beta_1e^{i\theta}\ket{1}\big),
        \end{equation}
        with $\alpha_a,\beta_a\in\mathbb{R}$ and $\varphi,\theta\in[0,2\pi)$,
        the relevant mean values for the optimization become
    	\begin{equation}
            \label{eq:min-mean}
    		\begin{split}
    			\braket{\phi^{+}}_\chi&=\frac{1}{2}\big|\alpha_0 \beta_0 + \alpha_1\beta_1e^{i(\varphi+\theta)}\big|^2\leq \frac{1}{2}\big|\alpha_0 \beta_0 + \alpha_1\beta_1\big|^2,\\
    			\braket{\psi^{+}}_\chi&=\frac{1}{2}\big|\alpha_0 \beta_1e^{i\theta} + \alpha_1\beta_0e^{i\varphi}\big|^2 \leq \frac{1}{2}\big|\alpha_0 \beta_1 + \alpha_1\beta_0\big|^2,\\
    			|\braket{+0|\chi}|^2&=\frac{1}{2}\big|\alpha_0 + \alpha_1e^{i\varphi}\big|^2|\beta_1|^2\leq \frac{1}{2}\big|\alpha_0 + \alpha_1\big|^2|\beta_1|^2.\\
    		\end{split}
    	\end{equation}
        As we wish to minimize the expression in Eq. \eqref{eq:min-rand-uni}, we have to simultaneously maximize the quantities in Eq. \eqref{eq:min-mean}.
        Thus, the relative phases $e^{i\varphi}$ and $e^{i\theta}$ can be neglected, justifying an optimization over real-valued amplitudes.
    \end{widetext}

\end{document}